\documentclass{article}
\usepackage{geometry}
\usepackage{mathptmx}

\newcommand{\E}{\mathrm{E}}
\newcommand{\Var}{\mathrm{Var}}

\newcommand{\svr}{{\mathcal B}_r(v)}
\newcommand{\svone}{{\mathcal B}_1(v)}
\newcommand{\svtwo}{{\mathcal B}_2(v)}

\newcommand{\hll}{{\sc HyperLogLog}\xspace}
\newcommand{\hb}{{\sc HyperBall}\xspace}
\newcommand{\eb}{{\sc HyperEdgeball}\xspace}
\newcommand{\hanf}{{\sc HyperANF}\xspace}
\newcommand{\vp}{Vysochanskij-Petunin\xspace}
\newcommand{\vol}{\text{vol}}
\newcommand{\edge}{|\mathcal{E}_r(v)|}
\newcommand{\edgedirected}{|\mathcal{E}^{-}_r(v)|}
\newcommand{\estedge}{\widehat{|{\mathcal{E}}_r(v)|}}
\newcommand{\estedgedirected}{\widehat{|\mathcal{E}^{-}_r(v)|}}
\newcommand{\esttriangle}{\hat{\Delta}_r(v)}
\renewcommand{\triangle}{\Delta_r(v)}
\newcommand{\etam}{\frac{\beta_p}{\sqrt{p}}}
\newcommand{\estwedge}{\hat{w}_r(v)}
\renewcommand{\wedge}{w_r(v)}
\newcommand{\phie}{\phi\big(S_r(v)\big)}
\newcommand{\estphi}{\hat{\phi}\big(\mathcal{B}_r(v)\big)}

\usepackage{natbib}
\bibpunct{[}{]}{,}{n}{,}{;}
\usepackage{algorithm}
\usepackage{algorithmicx}
\usepackage{algpseudocode}
\usepackage{xspace}
\usepackage{tikz}
\usepackage{subcaption}
\usepackage{amsmath,amsfonts,amssymb,amsthm}

\usetikzlibrary{patterns}
\usetikzlibrary{fadings}
\newtheorem{definition}{Definition}
\newtheorem{theorem}{Theorem}

\usepackage{url}

\newcommand{\footremember}[2]{%
   \footnote{#2}
    \newcounter{#1}
    \setcounter{#1}{\value{footnote}}%
}
\newcommand{\footrecall}[1]{%
    \footnotemark[\value{#1}]%
} 

\begin{document}

\title{Locating highly connected clusters in large networks with HyperLogLog counters}


\author{
Lotte Weedage\footremember{ut}{Department of Electrical Engineering, Mathematics and Computer Science, 
University of Twente, the Netherlands}, Nelly Litvak\footrecall{ut}, Clara Stegehuis\footrecall{ut}}

\maketitle

\begin{abstract}
{{ In this paper we introduce a new method to locate highly connected clusters in a network. Our proposed approach adapts the \hb algorithm \cite{boldi2013core} to localize regions with a high density of small subgraph patterns in large graphs in a memory-efficient manner. We use this method to evaluate  three measures of subgraph connectivity: conductance, the number of triangles, and transitivity. We demonstrate that our algorithm, applied to these measures, helps to identify clustered regions in graphs, and provides good seed sets for community detection algorithms such as PageRank-Nibble. We analytically obtain the performance guarantees of our new algorithms, and demonstrate their effectiveness in a series of numerical experiments on synthetic and real-world networks.}}
{hyperloglog, probabilistic counting, network clustering}
\\
\end{abstract}

\section{Introduction}
Networks describe the connections between pairs of objects. Examples of networks are ubiquitous and include social networks, the Internet or communication networks. While these examples are very different from an application point of view, they share many characteristics. For example, in many real-world networks, objects have the tendency to cluster together in groups.  The task of finding these densely connected spots, or clusters, in the network, has been  a subject of vast research. 

A quantity that is commonly used to measure the quality of clusters in a graph representation of a network, is the conductance of a cluster. The conductance is based on the \textit{min-cut} \cite{stoer1997simple}, and gives a ratio of the number of  edges connected to nodes outside the cluster relative to the number of edges inside the cluster. The use of conductance as a measure to find communities is one of the most useful and important cut-based methods that exists \cite{schaeffer2007graph}.  

Another measure that can find clustered groups of nodes in a graph is the number of  triangles, because a triangle in a graph is the most clustered subgraph consisting of three nodes. For any subgraph in a network, one can  measure its quality as a cluster by counting  the number of triangles in this subgraph. Another option is to find its \textit{transitivity}, which is the ratio of the number of triangles versus the number of wedges in the subgraph, and therefore it tells us how clustered the graph is. 
Moreover, the task of finding triangles itself  has interesting applications in for example, biological networks \cite{tran2013counting}, spam detection \cite{becchetti2008link} or link recommendations \cite{tsourakakis2011spectral}. 

Finding dense parts of a graph, that have a low conductance, a high number of triangles, or a high transitivity, is a computationally demanding task because real-world networks often contain millions or even billions of nodes. 
In this paper we propose new methods to efficiently compute three measures of clustering -- the conductance, the number of triangles, and the transitivity -- for the ball subgraphs $\svr$:  the induced subgraphs that contain all nodes within graph distance $r$ of node $v$. The identified ball subgraphs of low conductance or high transitivity reveal locations of  dense areas in the network and can be used in community detection, for example, as \textit{seed sets} of other more time- and memory-consuming algorithms  such as {\sc PageRank Nibble} \cite{andersen2006local} or the {\sc Multi Walker Chain} model \cite{bian2017many}.

Our proposed algorithms for computing conductance and transitivity use probabilistic HyperLogLog counters to estimate the number of edges, wedges and triangles in ball subgraphs. This class of randomized algorithms stems from the \hll algorithm \cite{flajolet2007hyperloglog}  for counting the number of distinct elements in large streams of data, such as the number of unique visitors on a web page or the number of different genomes in biological data. These counters give an accurate estimate of large cardinalities, and moreover are memory-efficient.  Boldi and Vigna \cite{boldi2011hyperanf} have successfully adapted the \hll algorithm to the networks context by developing the \hb algorithm, which counts the number of nodes in a ball subgraph $\svr$ for every node $v$ and every radius $r$. They used this to approximate centrality measures in large graphs and to find the distribution of distances between pairs of nodes in a network of Facebook users \cite{backstrom2012four}.  This idea of counting nodes in ball subgraphs has also been extended to counting distinct edges in ball subgraphs \cite{essay79108} or in a stream of edges \cite{wang2017approximately}. In this paper, we demonstrate that the potential of HyperLogLog-type counters on graphs is greater than only counting nodes and edges, but can extend to counting other patterns in networks, and can be used to approximate popular measures of clustering.

The main contribution of this paper is in designing memory-efficient HyperLogLog-based algorithms for computing  the conductance, the number of triangles, and the transitivity in ball subgraphs. We analytically derive accurate error bounds for these algorithms, and empirically confirm their high performance. Moreover, we demonstrate applications of our methods to community detection in synthetic and real-world networks. Our results show that the identified highly clustered ball subgraphs perform very well as seed sets of the {\sc PR-Nibble} algorithm \cite{andersen2006local}, improving on previously used benchmarks.

The structure of this paper is as follows. In Section~\ref{sec:preliminaries}, we provide a brief recap on algorithms for counting nodes and edges in graphs using HyperLogLog counters. In Section~\ref{sec:algorithm}, we extend these algorithms to other patterns in graphs, and we present our new methods for approximating the conductance, the number of triangles,  and the transitivity in ball subgraphs. In Section~\ref{sec:performance}, we analytically derive accurate error bounds of the estimators for the conductance, triangle count and transitivity. In Section ~\ref{sec:results}, we experimentally evaluate the performance of our algorithms on a number of synthetic networks. In Section~\ref{sec:com_detection}, we demonstrate application of our methods to community detection. We conclude in Section~\ref{sec:discussion} with discussion.

\section{Preliminaries: \hll, \hb, and \eb}
\label{sec:preliminaries}

\subsection{\hll algorithm}
The \hll algorithm is a probabilistic counting technique that estimates the number of distinct elements of a large dataset, called the dataset cardinality. While a naive deterministic approach requires storage of all distinct elements observed so far, the randomized \hll algorithm is extremely memory-efficient, yet delivers accurate cardinality estimations. 

We will now briefly outline the idea of the  \hll algorithm~\cite{flajolet2007hyperloglog}. 
 The input of the algorithm is a \textit{multiset} $\mathcal{M}$, a stream of data items that are read in order of occurrence.
The algorithm uses a hash function $h: \mathcal{M} \rightarrow \{0, 1\}^\infty$ that assigns a binary string to every element of $\mathcal{M}$. The hash function is deterministic in the sense that it assigns exactly the same value to identical elements of $\mathcal{M}$. However, $h$ is constructed in such a way that its bits can be assumed to be independent Bernoulli random variables with probability 1/2 of 0 and 1. Then one can use  the principle of \textit{bit-pattern observables}: for example, in order to encounter the pattern $00001$ at the beginning of a string, one needs to observe, on average, 32 different items. For such estimate, the algorithm needs to store in memory only `how rare' the `most rare' observed binary sequence is, for example, the maximal number of zeros observed so far at the beginning of a binary sequence. The name \hll refers to this extremely low, double-logarithmic, memory requirements. To obtain accurate estimates, the algorithm uses registers. That is, the first $b$ bits of $h$ are used for identifying one of the $p=2^b$ registers, and the string after that is used to compute the cardinality estimate in this register. The algorithm initialises an empty counter with $p = 2^b$ registers, where every register corresponds to an entry of the counter. The more registers we use, the more precise the cardinality estimate will be. More precisely, the \hll algorithm returns the following estimate $E$ of  the number of unique elements of multiset $\mathcal{M}$: 
\begin{align}
    E:= \frac{\alpha_p p^2}{\sum_{j=1}^p2^{-M[j]}}, \hspace{1cm} \text{with } \alpha_p := \Bigg(p \int_0^\infty \Bigg(\log_2\Big(\frac{2+u}{1+u}\Big)\Bigg)^p \text{d}u\Bigg)^{-1} \label{eq:hllsize},
\end{align}
where $M[j]$ is the cardinality estimate of register $j$.
The pseudocode of the \hll algorithm is provided in Algorithm~\ref{alg:hyperloglog} in Appendix \ref{sec:appendix_algorithms}. Overall, the \hll algorithm uses $(1 +o(1))p \log \log(n/p)$ bits of space~\cite{flajolet2007hyperloglog} for a set of cardinality $n$, making it an extremely memory-efficient algorithm to estimate cardinalities of large sets.

The expectation and the variance of $E$  are given in Theorem 1 of \cite{flajolet2007hyperloglog}, and will be used later in the paper for obtaining  performance guarantees of our algorithms:

\begin{definition}[Ideal multiset~\cite{flajolet2007hyperloglog}] An {\it ideal} multiset of cardinality $n$ is a sequence obtained by arbitrary replications and permutations applied to $n$ uniform identically distributed random variables over the real interval [0, 1].
\end{definition}

\begin{theorem}[from \cite{flajolet2007hyperloglog}]\label{thm:hyperloglogthm1}
Let the algorithm \textsc{Hyperloglog} be applied to an ideal multiset of (unknown) cardinality $n$, using $p \geq 3$ registers, and let $E$ be the resulting cardinality estimate. 
\begin{enumerate}
    \item[(i)] The estimate $E$ is asymptotically almost unbiased in the sense that, as $n \rightarrow \infty$,
        \begin{align}
            \frac{1}{n} \E(E) = 1 + \delta_1(n) + o(1), \text{ where }|\delta_1(n)| < 5 \cdot 10^{-5} \text{ as soon as }p \geq 2^4.
        \end{align}
    \item[(ii)] The standard error defined as $\frac{1}{n}\sqrt{\Var(E)}$ satisfies as $n \rightarrow \infty$,
        \begin{align}
            \frac{1}{n} \sqrt{\Var(E)} = \etam + \delta_2(n) + o(1), \text{ where }|\delta_2(n)| < 5 \cdot 10^{-4} \text{ as soon as }p \geq 2^4,
        \end{align}
    the constants $\beta_p$ being bounded, with $\beta_{16} = 1.106, \beta_{32} = 1.070, \beta_{64} = 1.054, \beta_{128} = 1.046,$ and $\beta_\infty = \sqrt{3\log(2) -  1} \approx 1.03896.$
\end{enumerate}
\end{theorem}

\subsection{\hb algorithm}\label{sec:hb_algorithm}
The \hb algorithm, introduced in \cite{boldi2013core}, estimates the size of the ball consisting of nodes within graph distance $r$ around a center node using HyperLogLog counters. The algorithm is an adaptation of the \hanf algorithm \cite{boldi2011hyperanf}, which is based on the fact that the nodeball around node $v$ with radius $r$, $\mathcal{B}_r(v)$, can be found iteratively:

\begin{definition}[Nodeball]\label{def:nodeball}
    The nodeball $\mathcal{B}_r(v)$ consists of every node in a ball of radius $r$ around node $v$. Define $\mathcal{B}_0(v) = \{v\}$. For $r>1$, define:
    \begin{align}
        \mathcal{B}_r(v) = \bigcup_{w: \{v, w\} \in E} \mathcal{B}_{r-1}(v) \cup \mathcal{B}_{r-1}(w).
    \end{align}
\end{definition}

The \hb algorithm uses one HyperLogLog counter per node and for each iteration, the counters of the neighbours of this node are added to the node's own counter. After each iteration $r$, the size of this counter is calculated, which equals the estimator of  $|{\mathcal{B}}_{r+1}(v)|$. By Theorem~\ref{thm:hyperloglogthm1}, every estimator is almost unbiased and has a relative error of at most $\beta_p/\sqrt{p}$, for $\beta_p < 1.046$ as soon as every HyperLogLog counter has $p > 128$ registers. This algorithm particularly excels at its small memory usage and the fact that the size of nodeballs around all nodes are found simultaneously. The pseudocode of the \hb algorithm is provided in Algorithm~\ref{alg:hyperball} {in Appendix} \ref{sec:appendix_algorithms}.

\subsection{\eb algorithm}
The \hb algorithm naturally extends from counting nodes to counting  edges that can be reached within radius $r$ around a node \cite{essay79108}. {The corresponding edgeballs are defined as follows:}

\begin{definition}[Edgeball]
    The edgeball $\mathcal{E}_0(v)$ consists of every edge incident to node $v$. Then, for $r > 1$:
    \begin{align}
    \label{eq:edgeball-1}
        \mathcal{E}_r(v) = \bigcup_{w: \{v, w\} \in E} \mathcal{E}_{r-1}(v) \cup \mathcal{E}_{r-1}(w).
    \end{align}
\end{definition}
Equivalently, we can rewrite (\ref{eq:edgeball-1}) as 
\begin{align}
    \label{eq:edgeball-2}
        \mathcal{E}_r(v) = \left\{\{x,y\}: x\in \svr\right\}.
    \end{align}

The \eb algorithm gives an approximation of the number of edges around a node after $r$ iterations, $|{\mathcal{E}}_{r+1}(v)|$. Compared to the \hb algorithm, the only change is in the initialisation phase (Algorithm \ref{alg:hyperball}), since we now count edges instead of nodes. This new initialisation is formalised in Algorithm~\ref{alg:hyperedgeball}. 

\begin{algorithm}[htb]
\begin{algorithmic}[1]
\State$c[-]$, an array of \textit{n} empty HyperLogLog counters
\State

\For{\textbf{each} $v \in V$} 
	\For{\textbf{each} $e \in \{\{v, w\} \in E\}$}
		\State \textsc{Add}$(c[v],w)$
	\EndFor
    \State write $\langle v, c[v] \rangle$ to disk
    \State \textbf{return} \textsc{Size}($c[v]$), which estimates $|\mathcal{E}_0(v)|$
\EndFor
\State
\State $(\widehat{|\mathcal{E}_r|})_{r\geq 1} = $ \textsc{CountBall}($c$)
\end{algorithmic}
\caption{The \eb algorithm. The \textsc{Add} and \textsc{Size} functions of Algorithm \ref{alg:hyperloglog} and the \textsc{CountBall} function of Algorithm \ref{alg:hyperball} are used to find the edgeball estimators.}
\label{alg:hyperedgeball}
\end{algorithm}
\section{Algorithms}\label{sec:algorithm}
We will now show that we can easily extend HyperBall-type algorithms to counting arbitrary, more complex patterns than edges or nodes. In particular, we will show that we can apply HyperBall-type algorithms to approximate three important clustering measures: conductance, triangles, and transitivity.

\subsection{Conductance}\label{sec:directed_edges}

The first quantity of interest that we investigate is conductance. The conductance of a graph $G$ is defined as follows:
\begin{definition}[Conductance]\label{def:conductance}
    For a graph $G = (V, E)$ with $n = |V|$ and $m = |E|$, the conductance of a subgraph $S \subset G$ is:
    \begin{align}
        \phi(S) = \frac{|\delta(S)|}{\min{(\vol(S), 2m - \vol(S))}},
    \end{align}
    where $	\delta(S) = \{\{x, y\} \in E| x \in S, y \notin S\}$ is the boundary of $S$, and $\vol(S)$ is  the volume of $S$ equal to the sum of the degrees of the nodes in subgraph $S$.
\end{definition}
In the rest of this paper we use a simplified definition of the conductance, as we assume that the graphs that we analyse are non-empty and the volume of the subgraphs that we analyze is smaller than the volume of its complement:
\begin{align}
    \phi\big(S\big) &= \frac{|\delta\big(S\big)|}{\vol\big(S\big)}.\label{eq:simplifiedconductance}
\end{align}

{Next, we wish to provide a memory-efficient way of estimating conductance of ball subgraphs. To this end, it seems natural to combine \hb and \eb from Section~\ref{sec:preliminaries}. However, this turns out to be insufficient because the low memory usage implies that we cannot identify the edges of $\delta(\svr)$. To overcome this problem, we transform our (undirected) graph into a directed variant of the graph where every undirected edge becomes two directed edges, and we introduce directed edgeballs:}

\begin{definition}[Out-edgeball]\label{def:out-edgeball}
    The out-edgeball with radius $r$ around node $v$ is defined as follows:
    \begin{align}
        \mathcal{E}_r^{-}(v) &= \{(x,y) \in E \hspace{0.05cm}\big|\hspace{0.05cm} x \in \mathcal{B}_r(v)\}.
    \end{align}
\end{definition}

Similarly, we can define an in-edgeball:
\begin{definition}[In-edgeball]\label{def:in-edgeball}
    The in-edgeball with radius $r$ around node $v$ is defined as follows:
    \begin{align}
        \mathcal{E}_r^{+}(v) &= \{(x,y) \in E \hspace{0.05cm}\big|\hspace{0.05cm} y \in \mathcal{B}_r(v) \}.
    \end{align}
\end{definition}

The different kinds of edgeballs are illustrated in Figure \ref{tikz:symmetricdirectedvsundirectededges}. For our purposes, it is important to notice that  when an edge is on the boundary between a node inside and a node outside the ball of radius $r$, this edge belongs only to the out-edgeball, while all other edges, that have both endpoints of the edge are in the nodeball $\mathcal{B}_r(v)$, belong to both in- and out-edgeball. This is exactly why the directed edgeballs are helpful for estimating conductance. 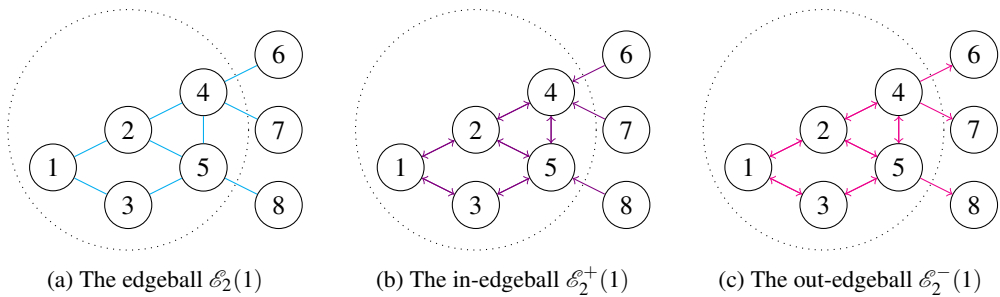
\begin{figure}[ht]
	\begin{subfigure}{0.30\textwidth}
    	\centering
    	\begin{tikzpicture}
    	\node[circle,draw, minimum size=0.5cm] (v11) at (-6.5,1) {$1$};
    	\node[circle,draw, minimum size=0.5cm] (v12) at (-5.5,1.5) {$2$};
    	\node[circle,draw, minimum size=0.5cm] (v13) at (-5.5,0.5) {$3$};
    	\node[circle,draw, minimum size=0.5cm] (v14) at (-3.5,1.5) {$7$};
    	\node[circle,draw, minimum size=0.5cm] (v15) at (-4.5,2) {$4$};
    	\node[circle,draw, minimum size=0.5cm] (v17) at (-4.5,1) {$5$};
    	\node[circle,draw, minimum size=0.5cm] (v18) at (-3.5,0.5) {$8$};
    	\node[circle,draw, minimum size=0.5cm] (v110) at (-3.5,2.5) {$6$};
    	
    	\draw[cyan]  (v11) edge (v12);
    	\draw[cyan]  (v12) edge (v15);
    	\draw[cyan]  (v12) edge (v17);
    	\draw[cyan]  (v13) edge (v17);
    	\draw[cyan]  (v11) edge (v13);
    	\draw[cyan]  (v15) edge (v17);
    	\draw[cyan]  (v15) edge (v110);
    	\draw[cyan]  (v15) edge (v14);
    	\draw[cyan]  (v17) edge (v18);
    	
    	\draw[dotted]  (-5.5,1.5) ellipse (1.6 and 1.6);
    \end{tikzpicture}
    \caption{The edgeball $\mathcal{E}_2(1)$}
    \label{tikz:edgeball}
    \end{subfigure}
    \begin{subfigure}{0.30\textwidth}
    	\centering
    	\begin{tikzpicture}
    	    	\node[circle,draw, minimum size=0.5cm] (v1) at (-6.5,1) {$1$};
    	\node[circle,draw, minimum size=0.5cm] (v2) at (-5.5,1.5) {$2$};
    	\node[circle,draw, minimum size=0.5cm] (v3) at (-5.5,0.5) {$3$};
    	\node[circle,draw, minimum size=0.5cm] (v7) at (-3.5,1.5) {$7$};
    	\node[circle,draw, minimum size=0.5cm] (v4) at (-4.5,2) {$4$};
    	\node[circle,draw, minimum size=0.5cm] (v5) at (-4.5,1) {$5$};
    	\node[circle,draw, minimum size=0.5cm] (v8) at (-3.5,0.5) {$8$};
    	\node[circle,draw, minimum size=0.5cm] (v6) at (-3.5,2.5) {$6$};
    	
    	\draw[violet, ->]  (v1) edge (v2);
    	\draw[violet, ->]  (v2) edge (v1);
    	\draw[violet, ->]  (v2) edge (v4);
    	\draw[violet, ->]  (v2) edge (v5);
    	\draw[violet, ->]  (v3) edge (v5);    	
    	\draw[violet, ->]  (v5) edge (v3);    	
	\draw[violet, ->]  (v5) edge (v4);    	
    	\draw[violet, ->]  (v5) edge (v2);
    	\draw[violet, ->]  (v1) edge (v3);
    	\draw[violet, ->]  (v3) edge (v1);
    	\draw[violet, ->]  (v4) edge (v5);
    	\draw[violet, ->]  (v4) edge (v2);
    	\draw[violet, ->]  (v6) edge (v4);
    	\draw[violet, ->]  (v7) edge (v4);
    	\draw[violet, ->]  (v8) edge (v5);
    	
    	\draw[dotted]  (-5.5,1.5) ellipse (1.6 and 1.6);
    \end{tikzpicture}
    \caption{The in-edgeball $\mathcal{E}^+_2(1)$}
    \label{tikz:in-edgeball}
    \end{subfigure}
    \begin{subfigure}{0.30\textwidth}
    	\centering
    	\begin{tikzpicture}
    	    	\node[circle,draw, minimum size=0.5cm] (v1) at (-6.5,1) {$1$};
    	\node[circle,draw, minimum size=0.5cm] (v2) at (-5.5,1.5) {$2$};
    	\node[circle,draw, minimum size=0.5cm] (v3) at (-5.5,0.5) {$3$};
    	\node[circle,draw, minimum size=0.5cm] (v7) at (-3.5,1.5) {$7$};
    	\node[circle,draw, minimum size=0.5cm] (v4) at (-4.5,2) {$4$};
    	\node[circle,draw, minimum size=0.5cm] (v5) at (-4.5,1) {$5$};
    	\node[circle,draw, minimum size=0.5cm] (v8) at (-3.5,0.5) {$8$};
    	\node[circle,draw, minimum size=0.5cm] (v6) at (-3.5,2.5) {$6$};
    	
    	\draw[magenta, ->]  (v1) edge (v2);
    	\draw[magenta, ->]  (v2) edge (v1);
    	\draw[magenta, ->]  (v2) edge (v4);
    	\draw[magenta, ->]  (v2) edge (v5);
    	\draw[magenta, ->]  (v3) edge (v5);    	
    	\draw[magenta, ->]  (v5) edge (v3);    	
	\draw[magenta, ->]  (v5) edge (v4);    	
    	\draw[magenta, ->]  (v5) edge (v2);
    	\draw[magenta, ->]  (v1) edge (v3);
    	\draw[magenta, ->]  (v3) edge (v1);
    	\draw[magenta, ->]  (v4) edge (v5);
    	\draw[magenta, ->]  (v4) edge (v2);
    	\draw[magenta, ->]  (v4) edge (v6);
    	\draw[magenta, ->]  (v4) edge (v7);
    	\draw[magenta, ->]  (v5) edge (v8);
    	
    	\draw[dotted]  (-5.5,1.5) ellipse (1.6 and 1.6);
    \end{tikzpicture}
    \caption{The out-edgeball $\mathcal{E}^-_2(1)$}
    \label{tikz:out-edgeball}
    \end{subfigure}
    \caption{Three methods for counting edges in an undirected graph.}
    \label{tikz:symmetricdirectedvsundirectededges}
\end{figure}
This is formally stated in the following  theorem:

\begin{theorem}\label{th:edgeboundary}
	For an undirected ball subgraph $S_r(v)$ it holds that:
	\begin{align}
		|\delta\big(\svr\big)| &= 2|\mathcal{E}_{r}(v)| - |\mathcal{E}^{-}_{r}(v)|,\label{eq:edgeboundarysubgraph}\\
			\text{vol}\big(\svr\big) & = |\mathcal{E}^{-}_r(v)| \label{eq:volumesubgraph}.
	\end{align}
\end{theorem}

\begin{proof}
Denote by ${\bf 1}\{A\}$ the indicator of $A$, and let $E'$ be the set of directed edges obtained from the edge set $E$ by replacing each undirected edge $\{x,y\}$ by two directed edges $(x,y)$ and $(y,x)$. By definition of the volume of a set of nodes, we obtain
\begin{align*}
\text{vol}\big(\svr\big) & =\sum_{x\in\svr}\sum_{y\in V}{\bf 1}\{\{x,y\}\in E\}= \sum_{x\in\svr}\sum_{y\in V}{\bf 1}\{(x,y)\in E'\} = |\mathcal{E}^{-}_r(v)|,  
\end{align*}
which proves (\ref{eq:volumesubgraph}). Next, using (\ref{eq:volumesubgraph}), we write
\begin{align*}
    |\delta\big(\svr\big)|&= \sum_{x\in \svr}\sum_{y\notin\svr}{\bf 1}\{\{x, y\} \in E\}\\
    &= \sum_{x\in \svr}\sum_{y\notin\svr}{\bf 1}\{\{x, y\} \in E\} +\text{vol}\big(\svr\big)-\text{vol}\big(\svr\big)\\
    &= \sum_{x\in \svr}\sum_{y\notin\svr}{\bf 1}\{\{x, y\} \in E\} +\sum_{x\in \svr}\sum_{y\in V}{\bf 1}\{\{x, y\} \in E\} -|\mathcal{E}^{-}_r(v)| \\
    &= 2\sum_{x\in \svr}\sum_{y\notin\svr}{\bf 1}\{\{x, y\} \in E\} +\sum_{x\in \svr}\sum_{y\in \svr}{\bf 1}\{\{x, y\} \in E\}-|\mathcal{E}^{-}_r(v)|. 
\end{align*}
The last expression equals to the right-hand side of (\ref{eq:edgeboundarysubgraph}) by the definition of $\mathcal{E}_{r}(v)$  and the fact that in the second term each undirected edge is counted twice.  \end{proof}
\medskip

Theorem \ref{th:edgeboundary} suggests the following estimator $\hat{\phi}\big(\svr\big)$ of  the conductance of $\svr$:
\begin{align}
	\hat{\phi}\big(\svr\big) = \frac{|{\delta}\big(\svr\big)|}{{\vol}\big(\svr\big)} = \frac{2\estedge - \estedgedirected}{\estedgedirected} = 2 \frac{\estedge}{\estedgedirected} - 1, \label{eq:conductanceestimate}
\end{align}
where $\estedge$ and $\estedgedirected$ are the estimates of $|\mathcal{E}_r(v)|$ and $|\mathcal{E}^-_r(v)|$, respectively. 

\begin{algorithm}[htb]
\begin{algorithmic}[1]
\State$c[-]$, an array of \textit{n} empty HyperLogLog counters
\State

\For{\textbf{each} $v \in V$} 
	\For{\textbf{each} $e \in \{(v, w) \in E'\}$}
		\State \textsc{Add}$(c[v],w)$
	\EndFor
	\State write $\langle v, c[v] \rangle$ to disk
    \State \textbf{return} \textsc{Size}($c[v]$), which estimates $|\mathcal{E}^-_0(v)|$
\EndFor
\State
\State $(\widehat{|\mathcal{E^-}_r|})_{r\geq 1} = $ \textsc{CountBall}($c$)
\end{algorithmic}
\caption{The directed \eb algorithm. The \textsc{Add} and \textsc{Size} functions of Algorithm \ref{alg:hyperloglog} and the \textsc{CountBall} function of Algorithm \ref{alg:hyperball} is used to find the out-edgeball estimators.}
\label{alg:directedhyperedgeball}
\end{algorithm}

\subsection{Triangles and wedges}
\label{sec:triangleandwedgecount}
We now present our algorithm that counts the number of triangles within a ball of radius $r$ around node $v$. We denote this number by ${\Delta}_r(v)$, and its estimator by $\hat{\Delta}_r(v)$. The idea is to obtain $\hat{\Delta}_r(v)$ using the same \hb algorithm as for counting edges or nodes, but we initialise the counter with triangles instead of nodes or edges. For this initialisation, we assign a unique hash value to each triangle in the graph. This can be done using algorithms for exact triangle counting such as \textit{compact-forward} \cite{latapy2006theory} or \textit{edge-iterator} \cite{batagelj2002generalized}. In Algorithm~\ref{alg:triangleballalgorithm}, we give an example of how this can be implemented.  

\begin{algorithm}[htb]
\begin{algorithmic}[1]
\State$c[-]$, an array of \textit{n} HyperLogLog counters
\State \textbf{triangle}, \textbf{wedge}: booleans that express whether triangles or wedges are counted
\State

\For{each $v \in V$} 
	\For{each $(i, v) \in E$}
	   \For{each $(j, v) \in E$}
	            \If{$(i, j) \in E$ \textbf{and triangle}}{}\label{line:triangleball}
	            \State \textsc{Add}$(c[v],(v, i, j))$
	            \ElsIf{\textbf{wedge}}
	            \State \textsc{Add}$(c[v],(v, i, j))$
	            \EndIf
	    \EndFor
	\EndFor
	\State write $\langle v, c[v] \rangle$ to disk
	\State \textbf{return} \textsc{Size}($c[v]$), which estimates $|\Delta_0(v)|$ or $|w_0(v)|$
\EndFor
\State
\If{\textbf{triangle}}{}
    \State $(\hat{\Delta}_r)_{r\geq 1} = $ \textsc{CountBall}($c$)
    \ElsIf{\textbf{wedge}}
    \State $(\hat{t}_r)_{r\geq 1} = $ \textsc{CountBall}($c$)
\EndIf
\end{algorithmic}
\caption{Triangle and wedge ball algorithm. The \textsc{Add} and \textsc{Size} functions of Algorithm \ref{alg:hyperloglog} and the \textsc{CountBall} function of Algorithm \ref{alg:hyperball} is used to find the triangle or wedge cardinality estimates.}
\label{alg:triangleballalgorithm}
\end{algorithm}

Now denote by $w_r(v)$ the number of wedges in $\mathcal{B}_r(v)$. Since wedges are open triangles, an estimator $\hat w_r(v)$ of $w_r(v)$ can be found in exactly the same way as $\hat\Delta_r(v)$, but in the initialisation we add a wedge $\{i,v,j\}$ to a counter if $i,j \in V$ are neighbours of $v$. (Note that in line \ref{line:triangleball} of Algorithm \ref{alg:triangleballalgorithm} we verify that $i$ and $j$ are neighbours; this is needed for the initialisation of the counter of triangles).

\subsection{Extension to counting of arbitrary induced subgraphs}
The algorithms above for counting triangles and wedges easily extend to counting arbitrary induced connected subgraphs, called {\it graphlets}, within balls $\mathcal{B}_r(v)$. For that, in the initialisation phase, we need to count the graphlets that involve node $v$ for all $v\in V$, and assign a unique hash value to each of these graphlets. After that, we can run the \hb algorithm to count graphlets in ball subgraphs. This approach potentially can indicate parts of networks with unusual quantities of particular graphlets. However, applying \hb to general subgraph patterns also comes with important difficulties.  First,  \hb  counts nodes/edges/graphlets in ball subgraphs, and therefore this approach cannot be easily extended to counting graphlets in subgraphs of any other form. Second, the initialisation phase presents a computational bottleneck because assigning a unique hash value to each graphlet supersedes the computationally demanding task of exact graphlet counting. 

\begin{algorithm}[htb]
\begin{algorithmic}[1]
\State$c[-]$, an array of \textit{n} HyperLogLog counters
\State

\For{each $v \in V$} 
   \If{$v \in $ graphlet}{}
	\State \textsc{Add}($c[v],$ graphlet)
	\State write $\langle v, c[v] \rangle$ to disk 
    \State \textbf{return} \textsc{Size}($c[v]$), which estimates the number of graphlets in $\svone$.
    \EndIf
\EndFor
\State
\State $($graphlet estimate$_r)_{r\geq 1}$ = \textsc{CountBall}($c$)
\end{algorithmic}
\caption{Graphlet ball algorithm. The \textsc{Add} and \textsc{Size} functions of Algorithm \ref{alg:hyperloglog} and the \textsc{CountBall} function of Algorithm \ref{alg:hyperball} is used to find the graphlet cardinality estimates.}
\label{alg:graphletballalgorithm}
\end{algorithm}

\section{Error bounds}\label{sec:performance}

\subsection{Error bounds for  the estimator of conductance}
We now introduce Theorem \ref{th:cheberrorbound} and \ref{th:vperrorbound} that give a lower and upper bound for the conductance estimator $\estphi$ from \eqref{eq:conductanceestimate} based on Chebyshev's inequality and \vp's inequality \cite{vysochanskij1980justification}.

\begin{theorem}[Chebyshev bound for the conductance estimator]\label{th:cheberrorbound}
For all $v\in V$, $r\ge 1$, the conductance estimator $\estphi$, as defined in \eqref{eq:conductanceestimate}, satisfies:
\begin{align}
    P\Bigg[\estphi \in \Bigg(\frac{1 - \epsilon}{1 + \gamma} \cdot \phie, \frac{1 + \epsilon}{1 - \gamma} \cdot \phie\Bigg)\Bigg] &\geq 1 - \eta^2 \Bigg(\frac{\edge^2}{p_1^2} + \frac{\edgedirected^2}{p_2^2}\Bigg),
\end{align}
with $\epsilon = \frac{p_1}{\edge} + \delta_1 + o_1\left(\edge\right)$, $\gamma = \frac{p_2}{\edgedirected} + \delta_1 + o_2\left(\edgedirected\right)$ and $p_1, p_2 > 0$, where $\delta_1 = 5 \cdot 10^{-5}$, and $o_1,o_2=o(1)$, as their argument goes to infinity.

\end{theorem}
The proof of this theorem is based on Theorem~\ref{thm:hyperloglogthm1}, and is given in Appendix~\ref{sec:proofs}.

When our estimators $\estedge$ and $\estedgedirected$ follow a unimodal distribution, we can also use \vp's (VP) inequality \cite{vysochanskij1980justification}:
\begin{theorem}[Vysochanskij-Petunin's inequality]\label{th:vpinequality}
    Assume that the random variable $X$ has a unimodal distribution with finite mean $\E(X)$ and variance $\Var(X) = \sigma^2$. Then, for $\lambda/\sigma > \sqrt{8/3}$,
    \begin{align}
        P\big(|X-\E(X)| \geq \lambda \big) \leq \frac{4\sigma^2}{9 \lambda^2}.
    \end{align}
\end{theorem}

By using this inequality instead of Chebyshev's inequality, we can obtain tighter error bounds than the ones in Theorem \ref{th:cheberrorbound}, given in the next theorem.

\begin{theorem}[\vp bound for the conductance estimator]\label{th:vperrorbound}
If $\estedge$ and $\estedgedirected$ have a unimodal distribution, then for all $v\in V$, $r\ge 1$, the conductance estimator $\estphi$, as defined in \eqref{eq:conductanceestimate}, satisfies:
\begin{align}
    P\Bigg[\estphi \in \Bigg(\frac{1 - \epsilon}{1 + \gamma} \cdot \phie, \frac{1 + \epsilon}{1 - \gamma} \cdot \phie\Bigg)\Bigg] &\geq 1 - \frac{4}{9}\eta^2 \Bigg(\frac{\edge^2}{\lambda_1^2} + \frac{\edgedirected^2}{\lambda_2^2}\Bigg),
\end{align}
with $\lambda_1 > \sqrt{8/3 \cdot \Var\big(\edge\big)}$, $\lambda_2 > \sqrt{8/3 \cdot  \Var\big(\edgedirected\big)}$, $\epsilon = \frac{\lambda_1}{\edge}  + \delta_1 + o_1\left(\edge\right)$ and $\gamma = \frac{\lambda_2}{\edgedirected} + \delta_1 + o_2\left(\edgedirected\right)$, where $\delta_1, o_1$ and $o_2$ are as in Theorem~\ref{th:cheberrorbound}.
\end{theorem}

The proof of this theorem is identical to the proof of Theorem \ref{th:cheberrorbound}, but it uses the \vp inequality instead of the Chebyshev's inequality. In Section \ref{sec:results} we will show that the VP bound indeed holds in our numerical experiments by using a statistical test of unimodality \cite{hartigan1985dip}. In future research it will be interesting to identify general conditions under which the HyperLogLog counters produce estimators with a unimodal distribution.

\subsection{Error bounds of the estimator for the triangle count}
We again use the Chebyshev inequality in order to find a lower and upper error bound for the triangle count estimator $\esttriangle$ from Algorithm~\ref{alg:triangleballalgorithm}:
\begin{theorem}[Chebyshev bound for the triangle count estimator]\label{th:cheberrorboundtriangles}
For all $v\in V$, $r\ge 1$, the triangle estimator $\esttriangle$ satisfies:
\begin{align}
    P\Big[\esttriangle \in \Big(\E\big(\esttriangle\big) - a, \E\big(\esttriangle\big)& + a\Big)\Big] \geq 1 - \frac{\eta^2 \triangle^2 }{a^2},
\end{align}
for $a > 0$ and $\eta = \etam + \delta_2 + o_1\left(\triangle\right)$, where $\delta_2 = 5 \cdot 10^{-4}$ and $o_1 = o(1)$, as its argument goes to infinity.
\end{theorem}
The proof of this theorem is again based on Theorem~\ref{thm:hyperloglogthm1}, and is given in Appendix~\ref{sec:proofs}.

The \vp inequality \cite{vysochanskij1980justification}, which holds whenever an estimator has a unimodal distribution, can also be used in order to get a slightly tighter error bound:
\begin{theorem}[\vp bound for the triangle estimator]\label{th:trianglevperrorbound}
If $\esttriangle$ has a unimodal distribution, then for all $v \in V, r \geq 1$, the triangle estimator $\esttriangle$ satisfies:
\begin{align}
    P\Bigg(\esttriangle \in \Big(\E\big(\esttriangle\big) - \lambda, \E\big(\esttriangle\big)& + \lambda\Big)\Bigg) \geq 1 - \frac{4 \eta^2 \triangle^2 }{9 \lambda^2},
\end{align}
for $\lambda > \sqrt{8/3 \cdot \Var\big(\triangle\big)}$ and $\eta$ as in Theorem~\ref{th:cheberrorboundtriangles}.
\end{theorem}

The proof of Theorem \ref{th:trianglevperrorbound} goes in the same way as the proof of Theorem \ref{th:cheberrorboundtriangles}, but with the \vp inequality instead of Chebyshev's inequality.

\subsection{Error bounds for transitivity}
 The transitivity of a graph $G$ is defined as follows:

\begin{definition}[Transitivity]\label{def:transitivity}
The transitivity of a graph $G$ equals to 
\begin{align}
    t(G) &= \frac{3\Delta(G)}{w(G)},
\end{align}
where $w(G)$ is the number of wedges,  and $\Delta(G)$ the number of triangles  in $G$.
\end{definition}

In order to find the transitivity of ball subgraphs, we need to find the number of wedges, $|w\big(\svr\big)|$, in these ball subgraphs, which we obtain by using Algorithm~\ref{alg:triangleballalgorithm}.

\subsubsection{Error bounds of the transitivity estimator}\label{sec:transitivity_error}
We can find the Chebyshev and \vp error bounds of our estimator for transitivity similarly to Theorem~\ref{th:cheberrorboundtriangles} and~\ref{th:trianglevperrorbound}. The transitivity estimator is
\begin{align}
    \hat{t}\big(\svr\big) := \frac{3\esttriangle}{\estwedge}. \label{eq:transitivity_estimator}
\end{align}
The expectation and variance of our estimators $\esttriangle$ and $\estwedge$ can be obtained from Theorem \ref{thm:hyperloglogthm1}, which results in the following Chebyshev error bound of the transitivity estimator:

\begin{theorem}[Chebyshev bound for the transitivity estimator]\label{th:transitivity_error_chebyshev}
For $v \in V, r \geq 1$, the transitivity estimator $\hat{t}\big(\svr\big)$, as defined in~\eqref{eq:transitivity_estimator}, satisfies:
\begin{align}
    P\Bigg[\hat{t}\big(\svr\big) \in \Bigg(\frac{1 - \epsilon}{1 + \gamma} \cdot t\big(\svr\big), \frac{1 + \epsilon}{1 - \gamma} \cdot t\big(\svr\big)\Bigg)\Bigg] &\geq 1 - \eta^2 \Bigg(\frac{\left(\Delta_r(v)\right)^2}{p_1^2} + \frac{\left(w_r(v)\right)^2}{p_2^2}\Bigg), 
\end{align}
for $p_1, p_2 > 0$ and $\epsilon = \frac{p_1}{\triangle} + \delta_1 + o_1\left(\triangle\right)$, $\gamma = \frac{p_2}{\wedge} + \delta_1 + o_2\left(\wedge\right)$, with $\delta_1 = 5 \cdot 10^{-5}$ and $o_1, o_2 = o(1)$ as their argument goes to infinity. 
\end{theorem}

When this transitivity estimate has a unimodal distribution amongst the nodes, we can again use the \vp inequality in order to get tighter error bounds:

\begin{theorem}[\vp bound for the transitivity estimator]\label{th:transitivity_error_vp}
If $\esttriangle$ and $\estwedge$ have a unimodal distribution, then for all $v \in V, r \geq 1$, the transitivity estimator $\hat{t}\big(\svr\big)$, as defined in \eqref{eq:transitivity_estimator}, satisfies:
\begin{align}
    P\Bigg[\hat{t}\big(\svr\big) \in \Bigg(\frac{1 - \epsilon}{1 + \gamma} \cdot t\big(\svr\big), \frac{1 + \epsilon}{1 - \gamma} \cdot t\big(\svr\big)\Bigg)\Bigg] &\geq 1 - \frac{4}{9}\eta^2 \Bigg(\frac{\Delta_r(v)^2}{\lambda_1^2} + \frac{w_r(v)^2}{\lambda_2^2}\Bigg),
\end{align}
with $\lambda_1 > \sqrt{8/3 \cdot \Var\big(\triangle\big)}$, $\lambda_2 > \sqrt{8/3 \cdot  \Var\big(\wedge\big)}$, $\epsilon = \frac{\lambda_1}{\triangle}  + \delta_1 + o_1\left(\triangle\right)$ and $\gamma = \frac{\lambda_2}{\wedge} + \delta_1 + o_2\left(\wedge\right)$, with $\delta_1, o_1$ and $o_2$ as in Theorem~\ref{th:transitivity_error_chebyshev}.
\end{theorem}

\section{Performance}\label{sec:results}
To evaluate the performance of our algorithms, we first run a series of experiments on artificial LFR graphs \cite{lancichinetti2008benchmark}. In LFR graphs, the nodes are divided into pre-defined communities, and one can choose  a \textit{mixing parameter} $\mu\in [0,1]$, which is the probability that an edge emanating from a node connects to a node in outside of its community. The LFR model involves a number of other parameters: the minimum and maximum community size $(|C_{min}|, |C_{max}|)$, the average and maximum degree $(\bar{d}, d_{max})$, the power-law exponent of the inverse cumulative distribution of  the node degrees ($\tau_1$), and the  power-law exponent of the inverse cumulative distribution of community sizes ($\tau_2$). We have used the LFR graph generator of NetworkX 2.4 in Python 3.6.9 to create LFR graphs with three different sets of  parameters as shown in Table \ref{tab:lfr-parameters}. We have used graphs with 1000 and 5000 nodes because in these small graphs we can find the exact number of edges, directed edges, wedges and triangles in all ball graphs, and  compare the performance of our algorithms to these exact results.

\begin{table}[ht]
\centering
\begin{tabular}{lllllllll}
\textbf{Graph} & \textbf{$n = |V|$} & \textbf{$\tau_1$} & \textbf{$\tau_2$} & \textbf{$|C_{min}|$} & \textbf{$|C_{max}|$} & \textbf{$\bar{d}$} & \textbf{$d_{max}$} & \textbf{$\mu$} \\ \hline
LFR-1 & 1000 & 2 & 3 & 10 & 50 & 10 & 50 & 0.3 \\
LFR-2 & 1000 & 2 & 3 & 20 & 100 & 20 & 100 & 0.3 \\
LFR-3 & 5000 & 2 & 3 & 10 & 50 & 10 & 50 & 0.3
\end{tabular}
\caption{Parameters for the generated LFR-graphs}
\label{tab:lfr-parameters}
\end{table}

For the experiments on LFR graphs, we have used a mixing parameter $\mu = 0.3$, since this was the smallest mixing parameter with no isolated nodes in every generated graph. We have investigated the conductance, number of triangles and transitivity in ball subgraphs of radii 1 and 2. A ball larger than this radius consists of a large part of the entire graph.

\subsubsection{Conductance}
Figure \ref{fig:conductance_lfr3} shows the exact and the estimated conductance in a LFR-3 graph in ball subgraphs of radius~1 (Figure \ref{fig:conductance_lfr3}a) and radius~2 (Figure \ref{fig:conductance_lfr3}b). On the horizontal axis, the nodes $v$ are arranged in the order of ascending conductance of $\svr$, $r=1,2$.   The DIP test for unimodality \cite{hartigan1985dip} gives a small $p-$value of $0.0085$, therefore we can reasonably assume that the conductance estimator is unimodal and apply the \vp error bounds (Theorem \ref{th:vperrorbound}).
Figure \ref{fig:conductance_lfr3} shows that the \vp bounds are tight and represent well the $95\%$-margin of the estimation.

\begin{figure}[ht]
    \begin{subfigure}{.49\textwidth}
    \centering
    \includegraphics[width = \textwidth]{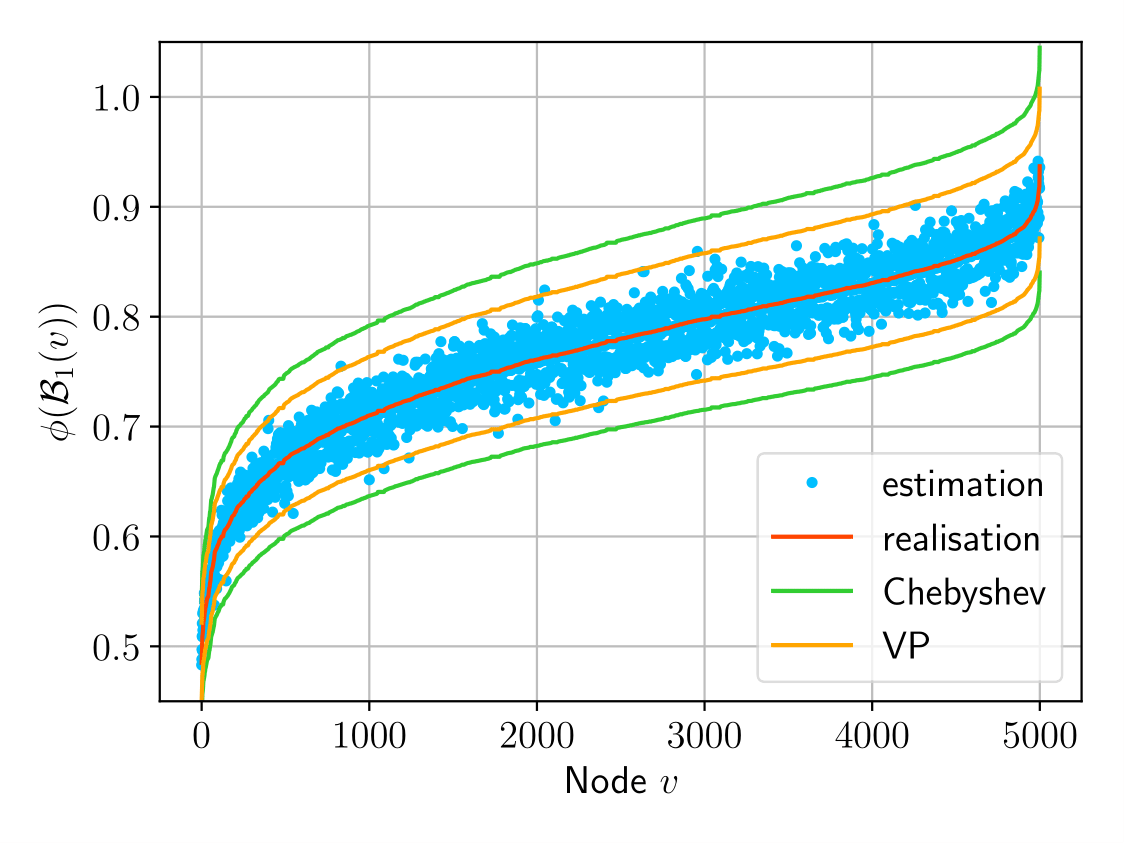}
    \caption{Conductance in $\svone$}
    \label{fig:conductance_lfr3-1}
    \end{subfigure}
    \begin{subfigure}{.49\textwidth}
    \centering
    \includegraphics[width = \textwidth]{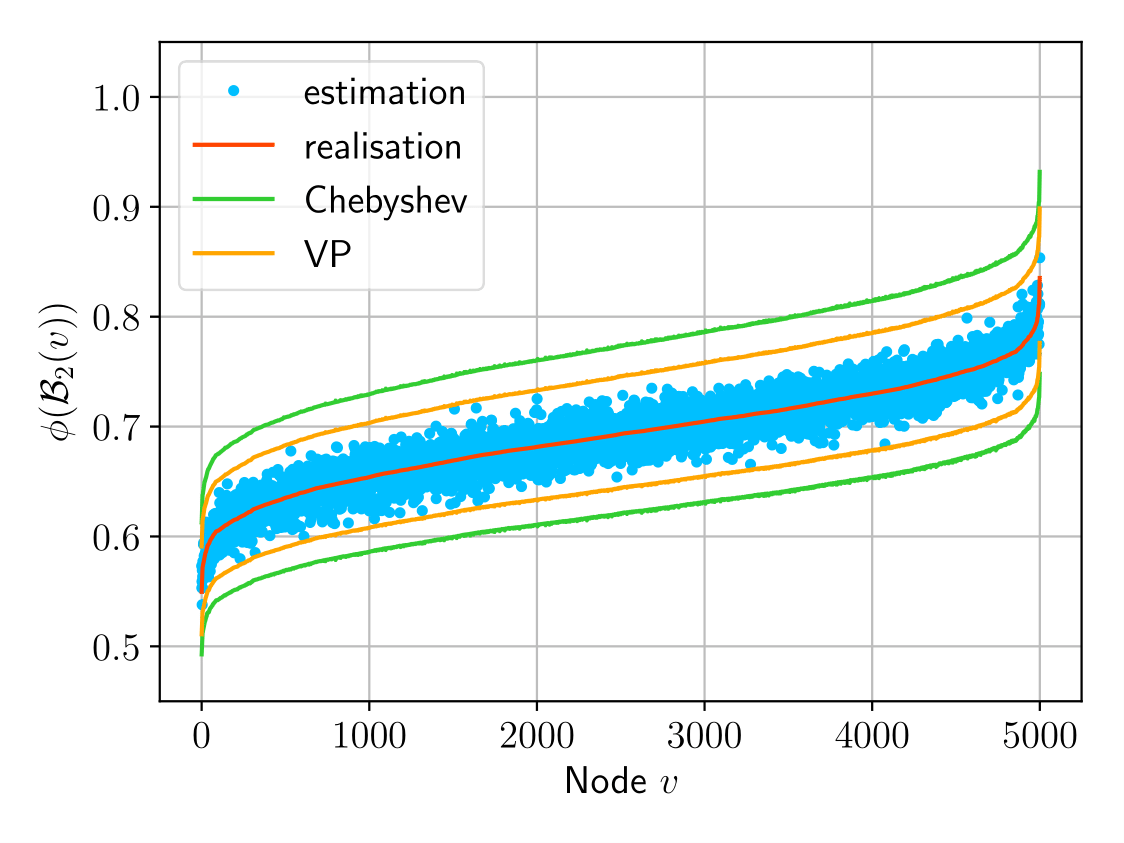}
    \caption{Conductance in $\svtwo$}
    \label{fig:conductance_lfr3-2}
    \end{subfigure}
    \caption{Estimated and exact conductance in the LFR-3 graph with $p = 2^{14}$ registers. The nodes $v$ are sorted by realised conductance of $\svr$ in ascending order.}
    \label{fig:conductance_lfr3}
\end{figure}

In order to investigate how the precision improves with increasing the number of registers, We have experimented with  different numbers of registers in the LFR-1 graph. The results are  shown in Table \ref{tab:lfr1-different-registers}. As we can see from this table, when the number of registers increases, the \vp and Chebyshev error bounds and the experimental error tighten rapidly, and the mean error decreases but keeps oscillating around 0, as expected from Theorem 1 of \cite{flajolet2007hyperloglog}. 

\begin{table}[ht]
\centering
\begin{tabular}{l|rrrrr}
\multicolumn{1}{c|}{$\mathbf{p}$} &
  \multicolumn{1}{c}{\textbf{VP}} &
  \multicolumn{1}{c}{\textbf{Chebyshev}} &
  \multicolumn{1}{c}{\textbf{Mean error}} &
  \multicolumn{1}{c}{\textbf{Variance error}} \\ \hline
$2^{8}$  & $\pm 0.7506$ & $\pm 1.367$ & $3.858 \cdot 10^{-3}$  & $0.01485$               \\ \hline
$2^{10}$ & $\pm 0.3219$ & $\pm 0.5230$ & $3.987 \cdot 10^{-4}$  & $3.236 \cdot 10^{-3}$   \\ \hline
$2^{12}$ & $\pm 0.1520$ & $\pm 0.2377$ & $2.945 \cdot 10^{-4}$  & $7.631 \cdot 10^{-4}$   \\ \hline
$2^{14}$ & $\pm 0.07560$ & $\pm 0.1155$ & $-8.756 \cdot 10^{-6}$ & $1.932 \cdot 10^{-4}$   \\ \hline
$2^{16}$ & $\pm 0.03929$ & $\pm 0.05954$ & $-6.172 \cdot 10^{-5}$ & $4.950 \cdot 10^{-5}$  \\ \hline
$2^{18}$ & $\pm 0.02158$ & $\pm 0.03285$ & $2.175 \cdot 10^{-5}$  & $1.183 \cdot 10^{-5}$ 
\end{tabular}
\caption{Average error bounds and experimental bounds for conductance in the ball subgraph $\svone$ in 100 generated LFR-1 graphs with different numbers of registers.}
\label{tab:lfr1-different-registers}
\end{table}

\subsection{Triangles}
In Figure \ref{fig:triangles_lfr} we show the number of triangles and their estimates in the LFR-3 graph in ball subgraphs of radii 1,2, and 3. Since the estimate for the number of triangles  again has a low $p-$value on the DIP test for unimodality $(p < 0.01)$, we can again use the \vp bounds. There is a large difference in the number of triangles between the balls of radius 2 and 3, which can be explained by the fact that ball subgraphs of radius 3 contain a large fraction of the entire graph.

\begin{figure}[ht]
    \centering
    \includegraphics[width = 0.45\textwidth]{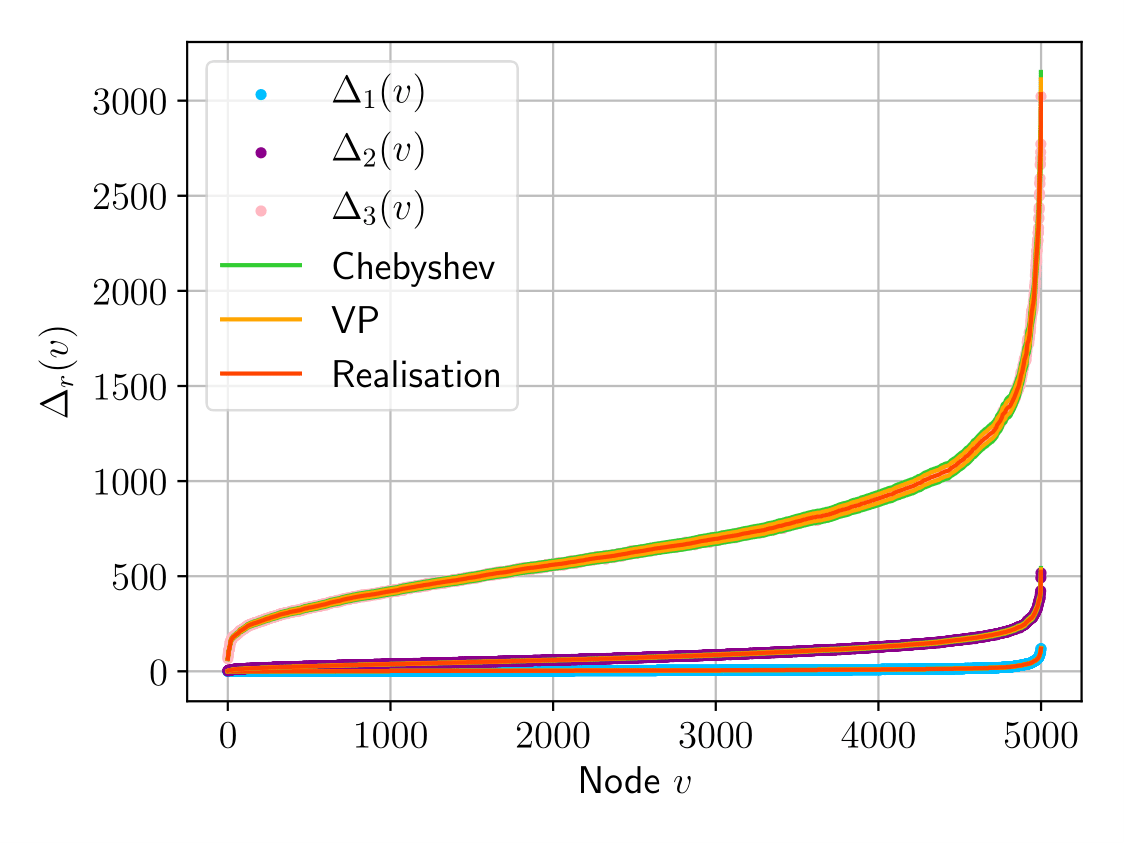}
    \caption{Estimate and realisation of the number of triangles $\Delta_r(v)$ in the ball subgraphs of radius $r=1,2,3$ in the LFR-3 graph; the number of registers is $p = 2^{14}$. The nodes $v$ are sorted by the realised number $\Delta_r(v)$ in ascending order.}
    \label{fig:triangles_lfr}
\end{figure}

\subsubsection{Transitivity}
In Figure \ref{fig:transitivity_lfr}, we show the transitivity in the ball subgraphs $\svone$ (Figure \ref{fig:transitivity_lfr-1}) and $\svtwo$ (Figure \ref{fig:transitivity_lfr-2}), together with the bounds from Theorem \ref{th:transitivity_error_chebyshev} and \ref{th:transitivity_error_vp}. 
\begin{figure}[ht!]
    \begin{subfigure}{0.45\textwidth}
    \centering
    \includegraphics[width = \textwidth]{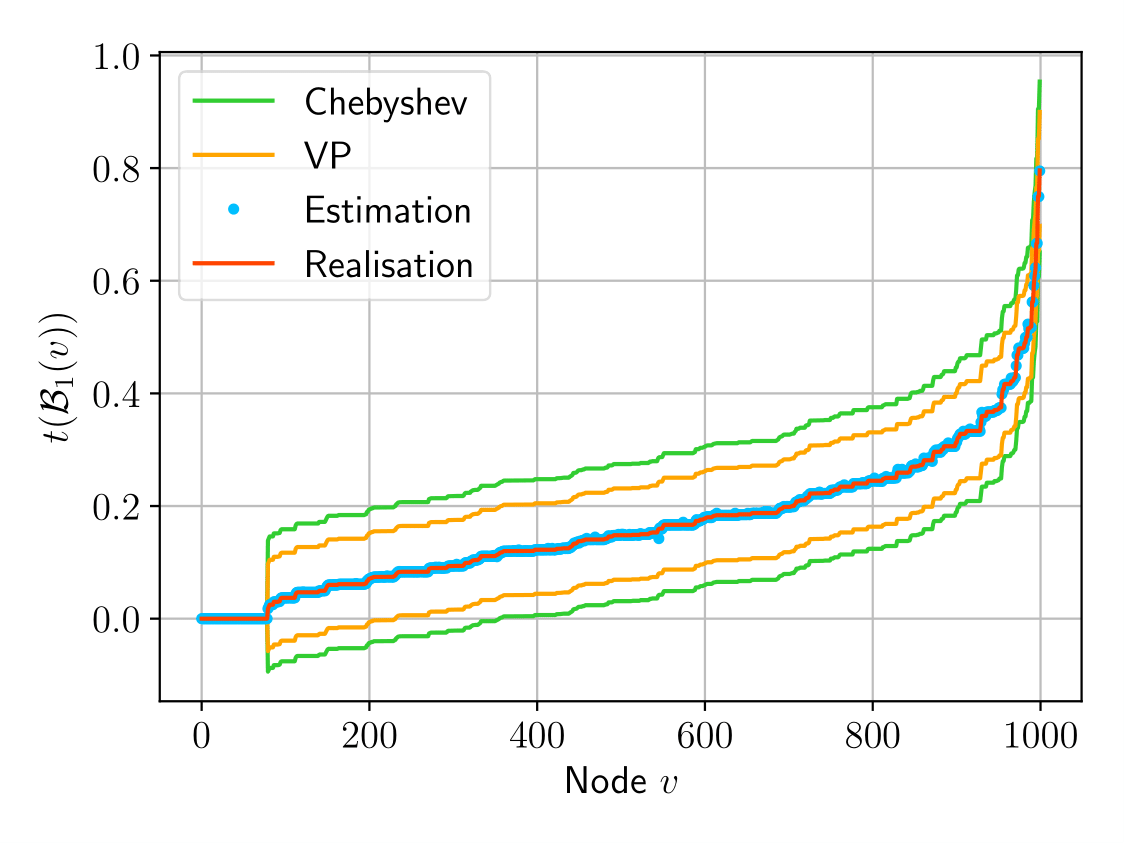}
    \caption{Transitivity in $\svone$}
    \label{fig:transitivity_lfr-1}
    \end{subfigure}
    \begin{subfigure}{0.45\textwidth}
    \centering
    \includegraphics[width = \textwidth]{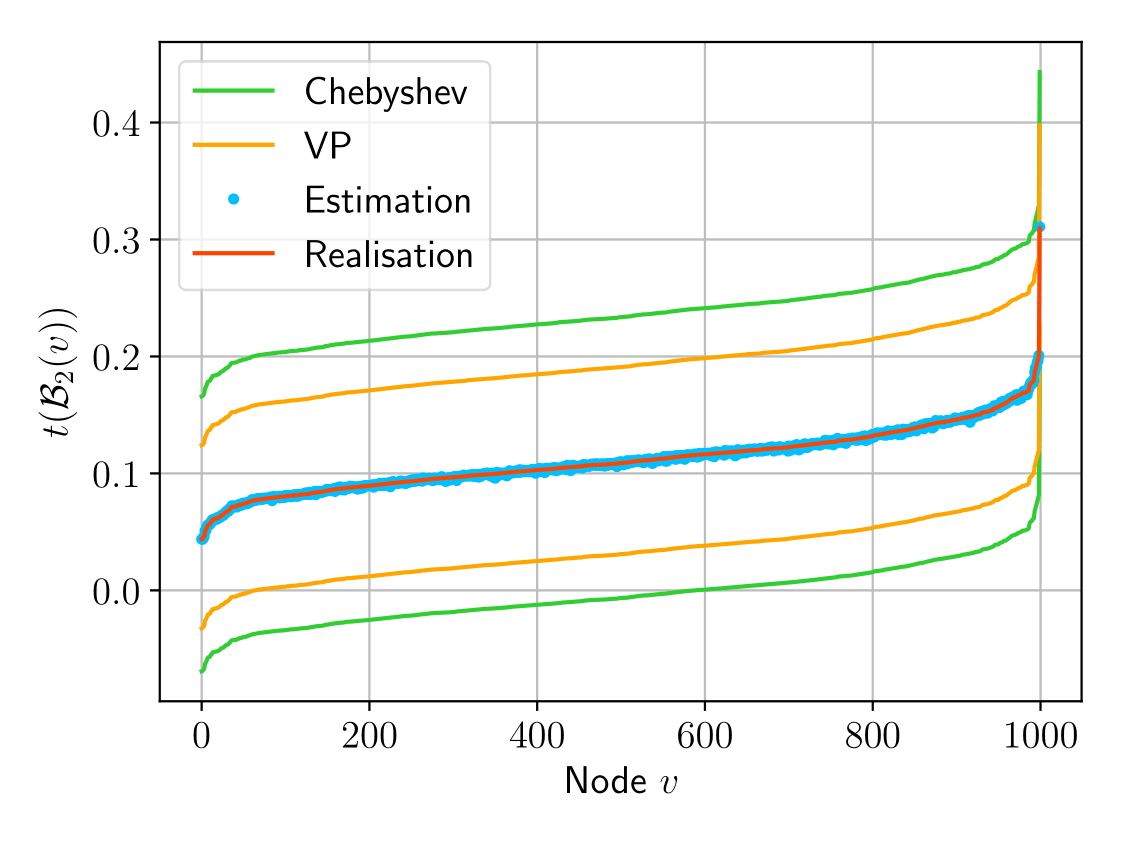}
    \caption{Transitivity in $\svtwo$}
    \label{fig:transitivity_lfr-2}
    \end{subfigure}
    \caption{Transitivity estimate in a LFR-1 graph with $p = 2^{14}$ registers. The nodes $v$ are sorted by realised transitivity in $\svr$, $r=1,2$, in ascending order.}
    \label{fig:transitivity_lfr}
\end{figure}

The DIP test for unimodality gives a $p-$value lower than $0.01$ so the \vp error bounds can again be used. Interestingly, the transitivity in ball subgraphs of radius 2 is already very small, which means that these ball subgraphs contain significantly more wedges than triangles. This is again an indication that the best communities in these kind of graphs are between the ball subgraphs of radius 1 and the ball subgraphs of radius 2, as also shown in \cite{gleich2012vertex}. Note that our estimation of the transitivity based on the \hb type algorithms is very precise, in fact the precision is much higher than suggested by the error bounds. This large precision may be explained by the fact that the number of wedges in a graph is often very large. This improves the precision of the transitivity estimator as well.

\section{Application to community detection}\label{sec:com_detection}
In this section we will show how our \hb-based algorithms can help to to detect communities in real-world networks. As an example, we use two large networks with a community structure: \textsc{com-DBLP} and \textsc{com-Amazon} \cite{snapnets, yang2015defining}. Table \ref{tab:real-life-graphs-properties} summarises the properties of these graphs.

\begin{table}[ht]
\centering
\begin{tabular}{l|llll}
\multicolumn{1}{c|}{\textbf{}} & \multicolumn{1}{c}{$n = |V|$} & \multicolumn{1}{c}{$|E|$} & \multicolumn{1}{c}{$|\Delta|$} \\ \hline
\textsc{com-DBLP}    & 317k & 1M   & 2M   \\ \hline
\textsc{com-Amazon}  & 334k & 925k & 667k 
\end{tabular}
\caption{Properties of the used real-world graphs}
\label{tab:real-life-graphs-properties}
\end{table}

We will identify the communities in these networks using the {\sc PageRank-Nibble} algorithm \cite{andersen2006local} that detects a community by finding a set of nodes with low conductance, starting from a random seed set of nodes. We propose to enhance this algorithm by using alternative seed sets found by our \hb algorithms. For example, a seed set can consist of ball subgraphs with the smallest conductance, or the largest density of triangles, or the largest transitivity. We have implemented this approach using 5 different seed sets, each of 100 nodes, in {\sc PageRank-Nibble}:
\begin{enumerate}
    \item {\it $\phi$ -seeds:} nodes $v$ with smallest conductance in $\svone$ and $\svtwo$;
    \item {\it $\Delta$-seeds:} nodes $v$ with largest number of triangles of radius 0 and radius 1, $|\Delta_0(v)|$ and $|\Delta_1(v)|$;
    \item {\it $t$-seeds:} nodes $v$ with highest transitivity in $\svone$ and $\svtwo$;
    \item {\it degree seeds:} nodes with highest degree;
    \item {\it random seeds:} randomly chosen nodes.
\end{enumerate}
For the {\sc PageRank-Nibble} algorithm, we have used a maximum cut size of 200, a teleport probability of $\alpha = 0.85$ and we calculate the $\epsilon$-approximate PageRank vector with $\epsilon = 10^{-8}$ \cite{andersen2006local}. We then compare the resulting conductance obtained with seed sets listed above.  

The results are presented in Figures \ref{fig:boxplotLFR_amazon} and \ref{fig:boxplotLFR_dblp}. For both graphs the lowest conductance subgraphs are found with the $\phi$-seeds. We notice that the found sets of small conductance are often small nearly isolated sets of nodes. Interestingly, using the $\phi$-seeds, {\sc PageRank-Nibble} is able to find such sets, while with degree seeds it fails to do so. 

In the Amazon graph, $t$-seeds result in low conductance subgraphs after {\sc PageRank-Nibble}, suggesting that the local high transitivity is a good indication for a community. Notice that random seeds also yield low conductance sets, but $t$-seeds clearly outperform this benchmark. 

\begin{figure}[ht]
    \begin{subfigure}{0.45\textwidth}
    \centering
    \includegraphics[width = \textwidth]{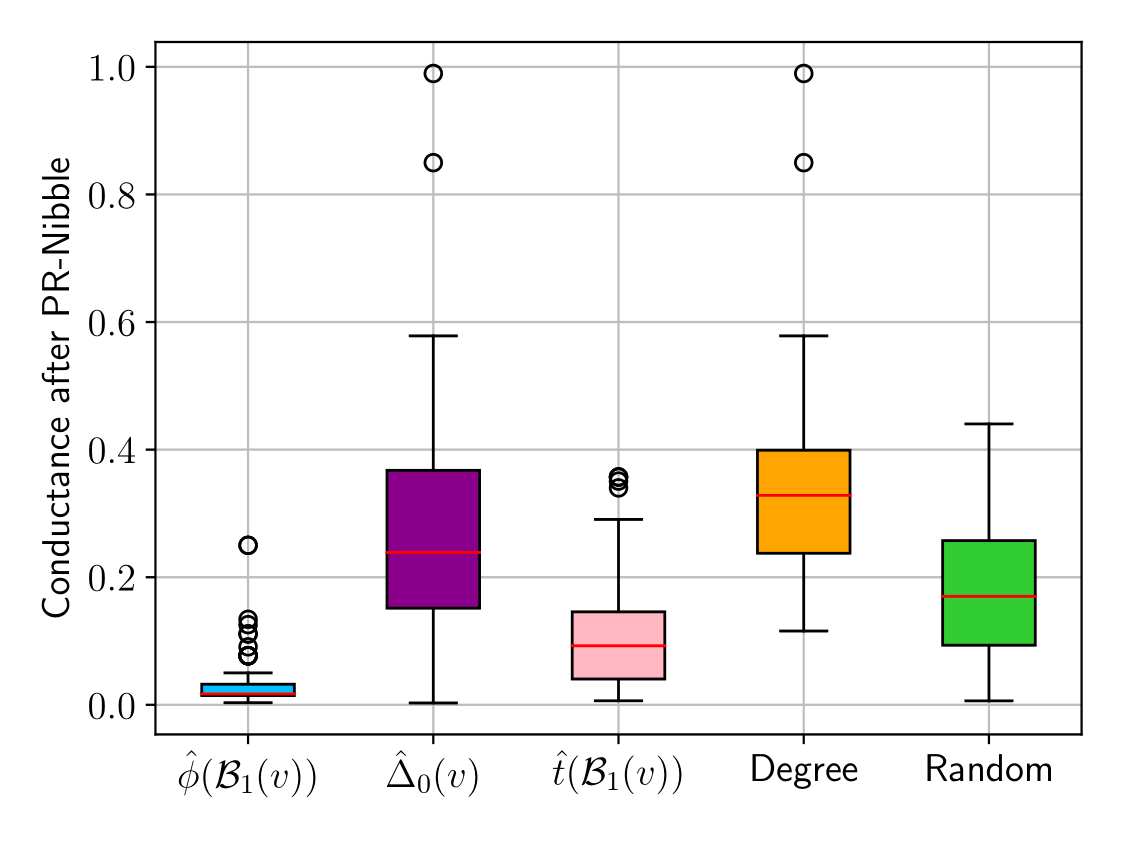}
    \caption{Ball subgraph of radius 1}
    \label{fig:boxplot_PR_nibble_amazon-1}
    \end{subfigure}
    \begin{subfigure}{0.45\textwidth}
    \centering
    \includegraphics[width = \textwidth]{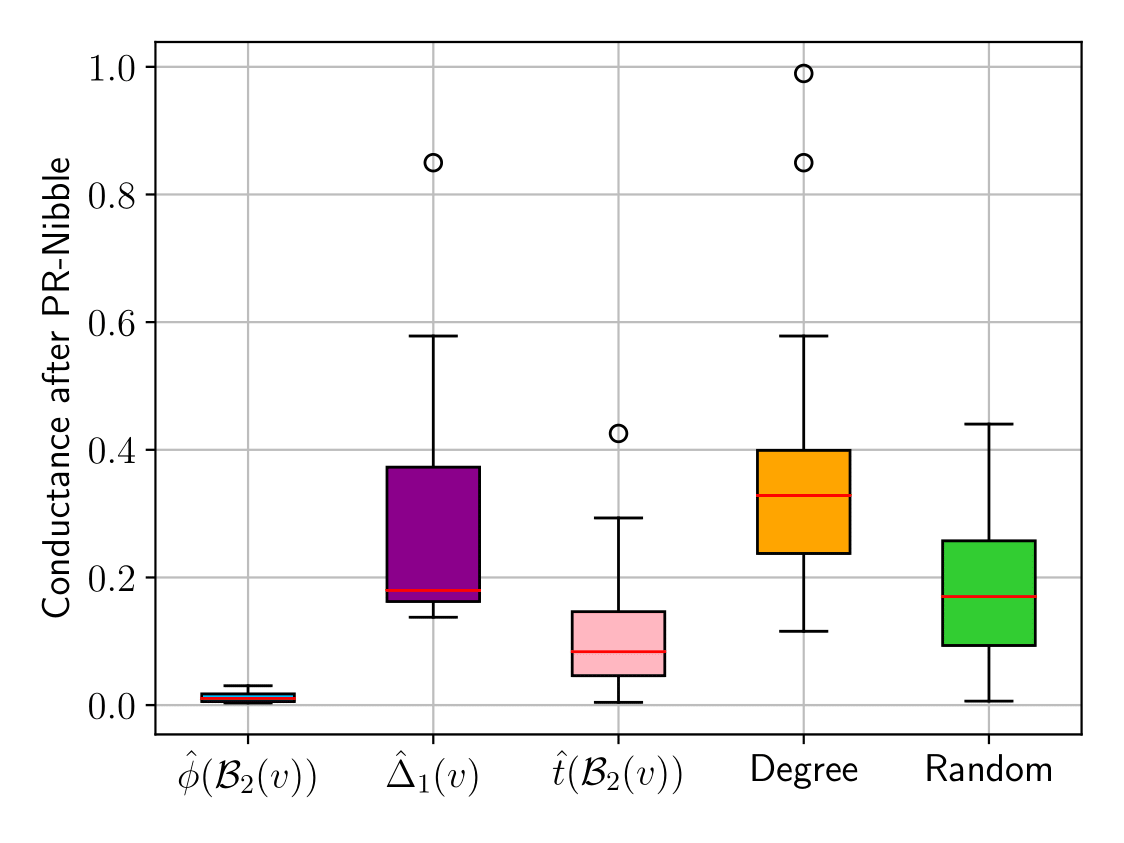}
    \caption{Ball subgraph of radius 2}
    \label{fig:boxplot_PR_nibble_amazon-2}
    \end{subfigure}
    \caption{Boxplots of the resulting conductance after PR-Nibble in every seed set in the Amazon graph}
    \label{fig:boxplotLFR_amazon}
\end{figure}

In the DBLP graph, Figure \ref{fig:boxplotLFR_dblp}, besides the $\phi$-seeds, both $\Delta$-seeds and $t$-seeds result in sets of low conductance after {\sc PageRank-Nibble}. Moreover, when seeds are chosen based on ball subgraphs of radius 2, $\Delta$-seeds  work the best. Not only this helps us to detect communities, but we also obtain an insight that as communities have high number and density of triangles, it is important to use this knowledge to detect them. Figure \ref{fig:boxplotLFR_dblp} confirms that it is an important feature for community detection in this network: the $t$-seeds and the $\Delta$-seeds  result in a much lower conductance after {\sc PageRank-Nibble} in comparison to the degree seeds and random seeds.

\begin{figure}[ht]
    \begin{subfigure}{0.45\textwidth}
    \centering
    \includegraphics[width = \textwidth]{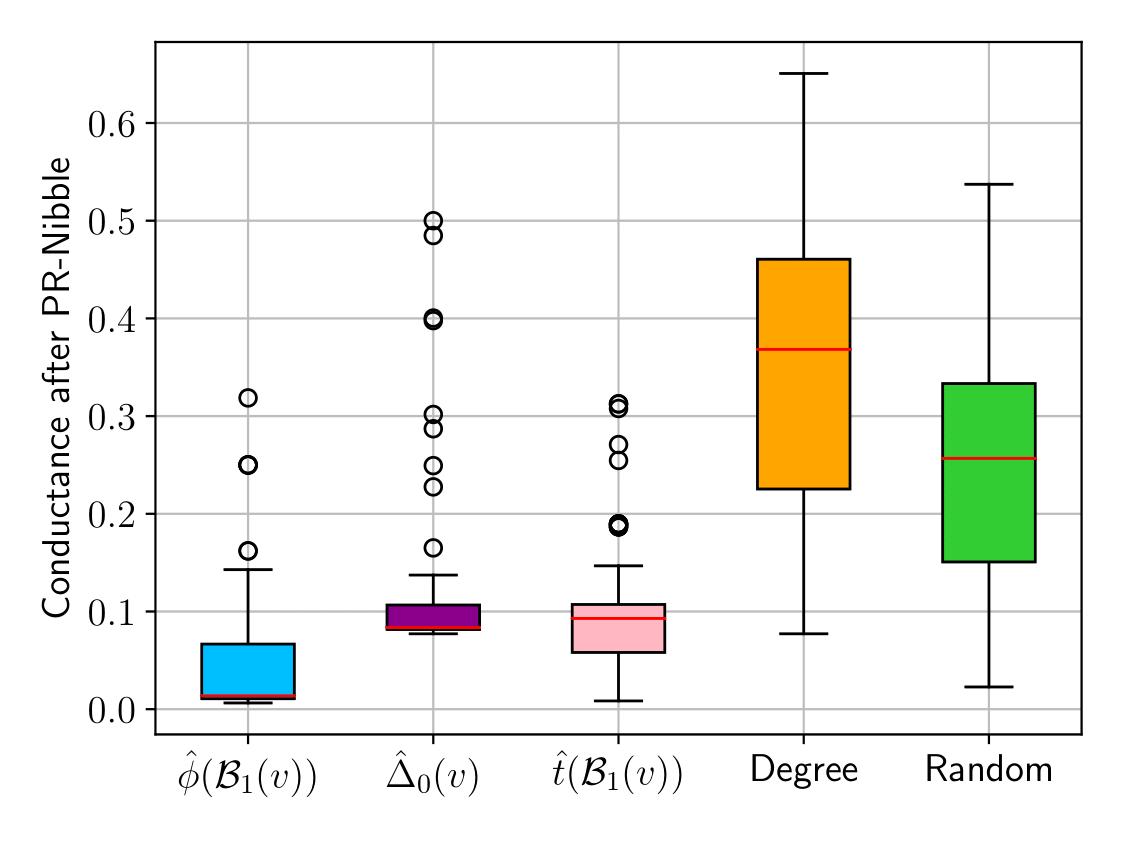}
    \caption{Ball subgraph of radius 1}
    \label{fig:boxplot_PR_nibble_dblp-1}
    \end{subfigure}
    \begin{subfigure}{0.45\textwidth}
    \centering
    \includegraphics[width = \textwidth]{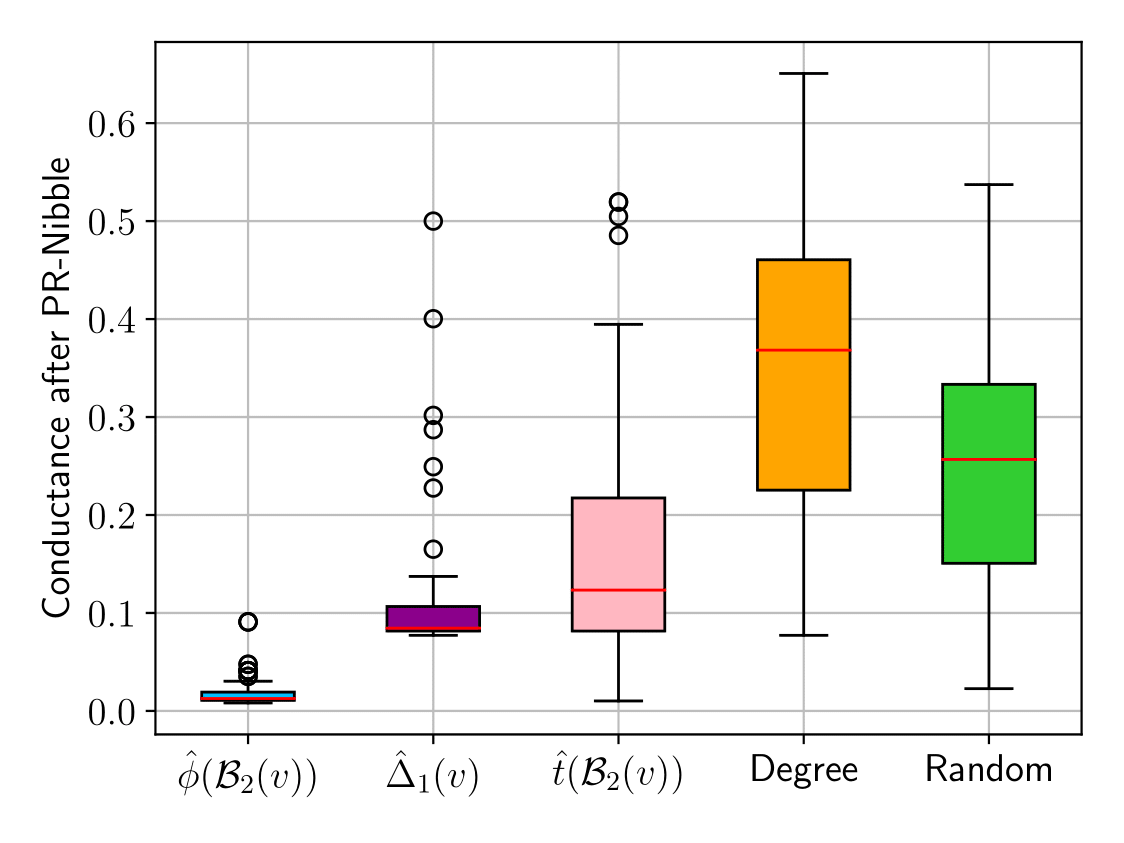}
    \caption{Ball subgraph of radius 2}
    \label{fig:boxplot_PR_nibble_dblp-2}
    \end{subfigure}
    \caption{Boxplots of the resulting conductance after PR-Nibble in every seed set in the DBLP graph}
    \label{fig:boxplotLFR_dblp}
\end{figure}
We conclude that while the performance of different seed sets differ in different real-world networks, our \hb-based convincingly outperforms the random and degree-based seed sets. 

\section{Discussion}\label{sec:discussion}
In this paper, we showed new applications of the \hb algorithm to count triangles and wedges. This enables us  to approximate statistics like the transitivity and the conductance in ball subgraphs. Our estimates have good precision, for which we have derived explicit bounds. In this paper we demonstrated how these algorithms can be applied to choose good seed sets for community detection and to understand in more detail the structure of  the communities. 

Moreover, we showed that \hb can be extended to count the number of graphlets of any form in ball subgraphs. This gives us means to find the areas in large networks with high concentration of graphlets of specific kind. Potentially this will yield new ways to find anomalities in networks~\cite{chen2013identification,becchetti2006link}, to distinguish the structure of networks of different nature (e.g. biological, technological or social networks)~\cite{milo2002network}, and to compare real-world networks to mathematical random graph models, where the results on most likely locations of particular graphlets have been explicitly derived in recent literature~\cite{vanderhofstad2020optimal,stegehuis2019variational}. The bottleneck of the \hb-type algorithms for local graphlet count is the initialisation that requires, for each node $v$, the exact count of the graphlets that involve $v$. How to resolve this bottleneck remains an interesting question for further research.

\bibliographystyle{plain}
\bibliography{sample}

\appendix
\counterwithin{algorithm}{section}
\section{HyperLogLog and HyperBall algorithm}\label{sec:appendix_algorithms}

\begin{algorithm}[H]
\begin{algorithmic}[1]
\State \textit{Let $h_b(x)$ be the first $b$ bits of the hashed value of element $x$}
\State \textit{Let $h^b(x)$ be the other part of the hashed value of  element $x$}
\State \textit{Let $\rho(h^b(x))$ be the position of the leftmost 1-bit ($\rho(001\cdots) = 3$).}
\State 
\State \textbf{initialise} a collection of $p = 2^b$ registers, $M[1], \ldots, M[p], to -\infty$.
\State
\Function{\textsc{Add}}{}{($M$:counter, $x$: item)}
	\State $i \leftarrow h_b(x)$
	\State $M[i] \leftarrow \max\{M[i], \rho(h^b(x))\}$
\EndFunction
\State
\Function{\textsc{Size}}{}{($M$: counter)}
	\State $Z \leftarrow \big( \sum_{j=0}^{p-1} 2^{-M[j]}\big)^{-1}$
	\State $E \leftarrow \alpha_p p^2 Z$
	\State
	\Return $E$
\EndFunction

\State
\For{\textbf{each} $x \in \mathcal{M}$}
	\State \textsc{Add}$(M, x)$
\EndFor
\State

\State \textbf{return} \textsc{Size}$(M)$
\end{algorithmic}

\caption{The \hll algorithm as described in \cite{flajolet2007hyperloglog}, which approximates the cardinality of a data stream $\mathcal{M}$.}
\label{alg:hyperloglog}
\end{algorithm}

\begin{algorithm}[H]
\begin{algorithmic}[1]
\State$c[-]$, an array of \textit{n} HyperLogLog counters initialised with nodes.
\State

\Function{\textsc{Union}}{}{($M$: counter, $N$: counter)}
\For{\textbf{each} $i < p$} 
\State $M[i] \leftarrow \max\{M[i], N[i]\}$ 
\EndFor
\EndFunction

\State

\Function{\textsc{CountBall}}{}{($c$: counter)}
\State $r \leftarrow 0$
\Repeat
\For{each $v \in V$}
	\State $a \leftarrow c[v]$
	\For{\textbf{each} $w \in \mathcal{N}(v)$}\label{alg:neighbourhood}
		\State $a \rightarrow$ \textsc{Union}$(c[w], a)$
	\EndFor
	\State write $\langle v, a\rangle$ to disk
\EndFor
\State update the array $c[-]$ with the new $\langle v, a\rangle$ pairs
\State $r \leftarrow r + 1$
\Until no counter changes its value
\State \Return \textsc{Size}(c)
\EndFunction

\State
\For{\textbf{each} $v \in V$} \Comment{Initialisation}\label{alg:start_initialisation}
\State \textsc{Add}$(c[v],v)$
\EndFor\label{alg:end_initialisation}
\State 

\State $(\widehat{|\mathcal{B}_r|})_{r\geq 1}$ = \textsc{CountBall}($c$)

\end{algorithmic}
\caption{The HyperBall algorithm as described in \cite{boldi2013core}, which returns an estimation of the ball cardinality for each node. The functions \textsc{Add} and \textsc{Size} of Algorithm \ref{alg:hyperloglog} are used.}
\label{alg:hyperball}
\end{algorithm}

\section{Proofs}\label{sec:proofs}

\subsection{Proof of Theorem \ref{th:cheberrorbound}}
To find the error bound of this estimator $\estphi$ we use Chebyshev's inequality. The estimator $\estphi$ is a fraction of two other estimator, $\estedge$ and $\estedgedirected$. The expectation and the variance of these estimators, as the cardinality of the estimated set goes to infinity, are given in Theorem \ref{thm:hyperloglogthm1}:
\begin{align}
    \E\big[\estedge\big] &\leq \edge \cdot \big(1 + \delta_1 + o(1)\big), \label{eq:expedge}\\
    \E\big[\estedgedirected\big] &\leq \edgedirected \cdot \big(1 + \delta_1 + o(1)\big),\\
    \Var\big[\estedge\big] &\leq \eta^2 \cdot \edge^2, \label{eq:varedge}\\
    \Var\big[\estedgedirected\big] &\leq \eta^2 \cdot \edge^2. \label{eq:varedgedirected}
\end{align}

The coefficient $\eta$ is defined as follows:
\begin{align}
\eta &= \etam + \delta_2 + o(1).\label{eq:eta3}
\end{align}
The upper bounds are derived by using the fact that $|\delta_1(x)| \leq 5 \cdot 10^{-5} = \delta_1$ for all $x$ and $|\delta_2(x)| \leq 5 \cdot 10^{-4} = \delta_2$ for all $x$, when the number of registers is larger or equal to $2^4$ (Theorem \ref{thm:hyperloglogthm1}). Throughout this work  we assume that the number of registers is always larger than $2^4$. 

Our goal is to use these estimators to find a lower and upper bound for $\estphi$. Following Chebyshev's theorem, we get the following inequalities for $p_1, p_2 > 0$ when the number of edges and directed edges in $\svr$ tend to infinity:
\begin{align}
    P\Big(\big|\estedge - \E\big[\estedge\big]\big| \geq p_1\Big) &\leq \frac{\edge^2 \cdot \eta^2}{p_1^2},\label{eq:cheb_edge}\\
    P\Big(\big|\estedgedirected - \E\big[\estedgedirected\big]\big| \geq p_2\Big) &\leq \frac{\edgedirected^2 \cdot \eta^2}{p_2^2}\label{eq:cheb_directed_edge},
\end{align}
with $\eta$ as defined in \eqref{eq:eta3}. We can rewrite the left-hand side of \eqref{eq:cheb_edge} using \eqref{eq:expedge}:
\begin{align}
    P\Big(\big|\estedge &- \E\big[\estedge\big]\big| \geq p_1\Big) = P\Big(\big|\estedge - \edge\big(1 + \delta_1\big(\edge\big) + o(1)\big)\big| \geq p_1\Big) \nonumber\\
    &= P\Bigg(\Bigg|\frac{\estedge - \edge\big(1 + \delta_1\big(\edge\big) + o(1)\big)}{\edge}\Bigg| \geq \frac{p_1}{\edge}\Bigg) \nonumber\\
    &= P\Bigg(\frac{\estedge}{\edge} \notin \Big(1 + \delta_1\big(\edge\big) + o(1) - \frac{p_1}{\edge}, 1 + \delta_1\big(\edge\big) + o(1) + \frac{p_1}{\edge}\Big)\Bigg)\nonumber\\
    &\leq P\Bigg(\frac{\estedge}{\edge} \notin (1 - \epsilon, 1 + \epsilon)\Bigg) \leq \frac{\edge^2 \cdot \eta^2}{p_1^2}, \label{eq:estimation_edge}
\end{align}
with $\epsilon = \frac{p_1}{\edge} + \delta_1 + o(1)$. The first inequality follows because $\delta_1(x) \leq \delta_1$ for all $x$, and the last inequality directly follows from \eqref{eq:cheb_edge}. Similarly, using \eqref{eq:cheb_directed_edge}, we obtain that for $\gamma = \frac{p_2}{\edgedirected} + \delta_1 + o(1)$,
\begin{align}
    P\Bigg(\frac{\estedgedirected}{\edgedirected} \notin (1 - \gamma, 
    1 + \gamma)\Bigg) &\leq \frac{\edgedirected^2 \cdot \eta^2}{p_2^2},\label{eq:estimation_directed_edge}
\end{align}
In order to find an error bound for the conductance estimate \eqref{eq:conductanceestimate} instead of the two estimators $\estedge$ and $\estedgedirected$, we define the following two events:
\begin{align}
    A &= \left\{\frac{\estedge}{\edge} \in  (1- \epsilon, 1 + \epsilon)\right\},\\
    B &= \left\{\frac{\estedgedirected}{\edgedirected} \in (1- \gamma, 1 + \gamma)\right\}.
\end{align}
By using the union bound and \eqref{eq:estimation_edge} and \eqref{eq:estimation_directed_edge}, we can express the intersection of events $A$ and $B$ as follows:
\begin{align}
    P \big(A \cap B\big) = 1 - 
    P\big(\bar{A} \cup \bar{B}\big) \geq 1 - P\big(\bar{A}\big) - P\big(\bar{B}\big)\geq 1 - \frac{\edge^2 \cdot \eta^2}{p_1^2} - \frac{\edgedirected^2 \cdot \eta^2}{p_2^2}. \label{eq:unionbound}
\end{align}
Now, we rewrite the event $A \cap B$ to obtain probabilistic bounds for $\estphi$: 
\begin{align}
    P\big(A \cap B\big) &\leq  P\Bigg[\frac{\frac{\estedge}{\edge}}{\frac{\estedgedirected}{\edgedirected}} \in \Bigg(\frac{1 - \epsilon}{1 + \gamma}, \frac{1 + \epsilon}{1 - \gamma}\Bigg)\Bigg]\nonumber\\
    &= P\Bigg[\frac{\frac{\estedge}{\estedgedirected}}{\frac{\edge}{\edgedirected}} \in \Bigg(\frac{1 - \epsilon}{1 + \gamma}, \frac{1 + \epsilon}{1 - \gamma}\Bigg)\Bigg]\nonumber\\
    &= P\Bigg[2\frac{\estedge}{\estedgedirected} - 1 \in \Bigg(\frac{1 - \epsilon}{1 + \gamma} \cdot \Big(2\frac{\edge}{\edgedirected} - 1\Big), \frac{1 + \epsilon}{1 - \gamma} \cdot \Big(2\frac{\edge}{\edgedirected} - 1\Big)\Bigg)\Bigg]\nonumber\\
    &= P\Bigg[\estphi \in \Bigg(\frac{1 - \epsilon}{1 + \gamma} \cdot \phie, \frac{1 + \epsilon}{1 - \gamma} \cdot \phie\Bigg)\Bigg].\label{eq:intersection}
\end{align}

Combining \eqref{eq:unionbound} and \eqref{eq:intersection} results in the error bounds for the conductance estimator:
\begin{align}
    P\Bigg[\estphi \in \Bigg(\frac{1 - \epsilon}{1 + \gamma} \cdot \phie, \frac{1 + \epsilon}{1 - \gamma} \cdot \phie\Bigg)\Bigg] &\geq 1 - \eta^2 \Bigg(\frac{\edge^2}{p_1^2} + \frac{\edgedirected^2}{p_2^2}\Bigg).
\end{align}

\subsection{Proof of Theorem \ref{th:cheberrorboundtriangles}}
\begin{proof}
The estimator of the triangle count of a graph comes directly from the Algorithm \ref{alg:hyperball}. This means that we know the expectation and variance of this estimator when the number of triangles goes to infinity (Theorem \ref{thm:hyperloglogthm1}):
\begin{align}
    \E\big[\esttriangle\big] &= 
    \triangle \cdot \Big(1 + \delta_1\big(\triangle\big) + o(1)\Big),\label{eq:triangleexpectation}\\
    \Var\big[\esttriangle\big] &= \triangle^2 \cdot \Big(\etam + \delta_2\big(\triangle\big) + o(1) \Big)^2,\label{eq:trianglevariance}
\end{align}
for two oscillating functions $\delta_1(x)$ and $\delta_2(x)$ with $|\delta_1(x)| < 5 \cdot 10^{-5} = \delta_1$ and $|\delta_2(x)| < 5 \cdot 10^{-4} = \delta_2$ as soon as $p \geq 16$ registers and $\beta_p \approx 1.03896$ for $p > 128$ registers. 

The proof of this theorem is straightforward: we use in Chebyshev's inequality for the triangle estimator, where we substitute the expectation and variance from Equations \eqref{eq:triangleexpectation} and \eqref{eq:trianglevariance}:
\begin{align}
    P\Big(\big| \esttriangle - \E\big(\esttriangle\big)\big| \geq a\Big) &\leq \frac{\Var\big(\esttriangle\big)}{a^2},\\
    P\Big(\big| \esttriangle - \E\big(\esttriangle\big)\big| \geq a\Big) &\leq \frac{\triangle^2 \cdot \Big(\etam + \delta_2\big(\triangle\big) + o(1) \Big)^2}{a^2},\\
    P\Big(\esttriangle \in \Big(\E\big(\esttriangle\big) - a, \E\big(\esttriangle\big)& + a\Big) \geq 1 - \frac{\Big(\etam + \delta_2\big(\triangle\big) + o(1) \Big)^2 \triangle^2 }{a^2},
\end{align}
for $a > 0$. Since $\eta = \etam + \delta_2 + o(1) \geq \etam + \delta_2\big(\triangle\big) + o(1)$, as defined in \eqref{eq:eta3}, this concludes the proof.
\end{proof}

Theorems \ref{th:transitivity_error_chebyshev} and \ref{th:transitivity_error_vp} are proved in the same way as Theorems \ref{th:cheberrorbound} and  \ref{th:vperrorbound}, since in both cases we need to find error bounds of a ratio of two estimators.

\end{document}